\documentclass[journal]{IEEEtran} 

\usepackage[colorlinks=true, allcolors=blue]{hyperref}

\usepackage[utf8]{inputenc} 
\usepackage[T1]{fontenc}    
\usepackage{hyperref}       
\usepackage{url}            
\usepackage{booktabs}       
\usepackage{amsfonts}       
\usepackage{nicefrac}       
\usepackage{microtype}      
\usepackage{xcolor}         

\usepackage{hyperref}
\usepackage{amsmath, mathtools, amsfonts, bm, amssymb, amsthm,graphicx, caption, subcaption,multirow}
\usepackage{algorithm} 

\usepackage{algpseudocode}

\usepackage{bm, bbm}
\usepackage{color}

\begin{document}

\newcommand{\x}{\bm{x}}
\newcommand{\g}{\bm{g}}
\newcommand{\G}{\bm{G}}
\newcommand{\X}{{\bm{X}}}
\newcommand{\s}{\bm{s}}
\renewcommand{\S}{\mathcal{S}}
\newcommand{\vv}{\bm{v}}
\renewcommand{\c}{c}

\newcommand{\e}{\bm{e}}
\newcommand{\y}{\bm{y}}
\newcommand{\w}{\bm{w}}
\newcommand{\U}{{\bm{U}}}
\newcommand{\R}{{\bm{R}}}
\newcommand{\M}{\bm{M}}
\renewcommand{\H}{\bm{H}}
\newcommand{\calG}{\mathcal{G}}

\newcommand{\I}{\bm{I}}
\newcommand{\Y}{\bm{Y}}
\newcommand{\uhat}{{\bm{\u}}}
\newcommand{\Zhat}{{\bm{\hat{Z}}}}
\newcommand{\Z}{{\bm{{Z}}}}
\newcommand{\D}{{\bm{D}}}
\newcommand{\DD}{{\bm{F}}}
\newcommand{\Q}{{\bm{Q}}}
\newcommand{\F}{{\bm{F}}}
\newcommand{\iidsim}{\stackrel{\mathrm{iid}}{\thicksim }}
\newcommand{\n}{{\cal{N}}}
\newcommand{\indepsim}{\stackrel{\mathrm{indep.}}{\thicksim }}
\newcommand{\SE}{\mathrm{SD}}  
\newcommand{\dist}{\mathrm{dist}}
\newcommand{\tta}{\ta^{\text{trunc}}}
\newcommand{\A}{\bm{A}}
\newcommand{\full}{{\mathrm{full}}}

\newcommand{\sg}{{\scalebox{.6}{{(g)}}}}
\newcommand{\sgp}{{\scalebox{.6}{{(g')}}}}
\newcommand{\sone}{{\scalebox{.6}{{(1)}}}}
\newcommand{\stwo}{{\scalebox{.6}{{(2)}}}}
\newcommand{\sL}{{\scalebox{.6}{{(L)}}}}
\newcommand{\con}{\mathrm{con}}
\newcommand{\epsfin}{\eps_{fin}}

\newcommand{\sj}{{\scalebox{.6}{{(j)}}}}
\newcommand{\tildeb}{\tilde{b}}
\newcommand{\gradU}{\mathrm{GradU}} 
\newcommand{\hatgradU}{\widehat{\mathrm{GradU}}}
\newcommand{\conserr}{\mathrm{ConsErr}}
\newcommand{\err}{\mathrm{Err}}
\newcommand{\Uerr}{\mathrm{UErr}}

\newcommand{\Span}{\mathrm{span}}
\newcommand{\rank}{\mathrm{rank}}
\newcommand{\evdeq}{\overset{\mathrm{EVD}}=} 
\newcommand{\svdeq}{\overset{\mathrm{SVD}}=} 
\newcommand{\qreq}{\overset{\mathrm{QR}}=} 
\newcommand{\bi}{\begin{itemize}} \newcommand{\ei}{\end{itemize}}
\newcommand{\ben}{\begin{enumerate}} \newcommand{\een}{\end{enumerate}}
\newcommand{\vs}{\vspace{0.1in}}
\newcommand{\vsl}{\vspace{0.05in}}
\newcommand{\vsm}{\vspace{-0.1in}}

\renewcommand{\implies}{\Rightarrow}


\newcommand{\cblue}{} 
\newcommand{\cbl}{\color{black}}
\newcommand{\cred}{}
\newcommand{\skipit}{ }

\newcommand{\trace}{\mathrm{trace}}

\newcommand{\qfull}{q_\full}
\newcommand{\sub}{{\mathrm{sub}}}

\newcommand{\bbf}{\mathbb{F}}
\newcommand{\dsW}{\mathds{W}}
\newcommand{\dsZ}{\mathds{Z}}

\newcommand{\mtx}[1]{\mathbf{#1}}
\newcommand{\vct}[1]{\mathbf{#1}}
\newcommand{\abs}[1]{\left|#1\right|}
\newcommand{\p}{\bm{p}}
\renewcommand{\j}{\bm{j}}
\newcommand{\uu}{{\bm{u}}}
\newcommand{\tot}{\mathrm{tot}}
\renewcommand{\a}{\bm{a}}
\newcommand{\ta}{\bm{\tilde{a}}}
\newcommand{\h}{\bm{h}}
\renewcommand{\b}{\bm{b}}
\renewcommand{\d}{\bm{d}}
\newcommand{\B}{\bm{B}}
\newcommand{\V}{{\bm{V}}}
\newcommand{\W}{{\bm{W}}}

\renewcommand{\aa}{\bm{a}}

\newcommand{\bP}{{\bm{P}}}

\newcommand{\z}{\bm{z}}
\newcommand{\zstar}{\z^\star}
\newcommand{\indic}{\mathbbm{1}}
\newcommand{\one}{\mathbm{1}}

\newcommand{\C}{\bm{C}}
\newcommand{\Chat}{\bm{\hat{C}}}
\newcommand{\cb}{\bm{c}}

\newcommand{\bz}{\boldsymbol{z}}

\newcommand{\tc}{\tilde{c}}
\newcommand{\tC}{\tilde{C}}

\setlength{\arraycolsep}{0.01cm}

\newcommand{\Bstar}{{\B^\star}}   
\newcommand{\bstar}{\b^\star}             
\newcommand{\dstar}{\d^\star}             

\newcommand{\tb}{\rho}

\newcommand{\Vcheck}{\check{\V}}
\newcommand{\Bcheck}{\check{\V}}
\newcommand{\bcheck}{\check{\v}}

\newcommand{\xhat}{\hat\x}
\newcommand{\Bhat}{\hat\B}
\newcommand{\Xhat}{\hat\X}

\newcommand{\bhat}{\hat\b}
\newcommand{\Uhat}{\hat\U}

\newcommand{\Utilde}{\widetilde\U}
\newcommand{\td}{\tilde{\bm{d}}^\star}
\newcommand{\init}{{\mathrm{init}}}

\newcommand{\Ustar}{\U^\star{}}
\newcommand{\Xstar}{\X^\star{}}
\newcommand{\xstar}{\x^\star}
\newcommand{\sstar}{\s^\star}
\newcommand{\Sstar}{\S^\star}
\newcommand{\Vstar}{\V^\star{}}

\newcommand{\deltinit}{\delta_\init}
\newcommand{\deltapt}{\delta_{t}}
\newcommand{\deltaptplus}{\delta_{t+1}}

\newcommand{\bSigma}{{\bm\Sigma^*}}
\newcommand{\tSigma}{\bm{E}_{det}}
\newcommand{\sigmin}{{\sigma_{\min}^\star}}
\newcommand{\sigmax}{{\sigma_{\max}^\star}}

\newcommand{\ik}{{ki}}
\newcommand{\J}{\mathcal{J}}

\renewcommand{\P}{\bm{P}}
\newcommand{\proj}{\mathcal{P}}
\newcommand{\E}{\mathbb{E}}
\newcommand{\norm}[1]{\left\|#1\right\|}

\renewcommand{\P}{\bm{U}}
\newcommand{\Phat}{\hat\P} 
\newcommand{\Lam}{\bm\Lambda} 
\renewcommand{\L}{\bm{L}}
\renewcommand{\V}{\bm{V}}

\renewcommand{\l}{\bm{\ell}}
\renewcommand{\v}{\bm{v}}
\newcommand{\tty}{\tilde\y}

\newcommand{\lhat}{\hat\l}

\newcommand{\at}{\a_t}
\newcommand{\yt}{\y_t}
\newcommand{\lt}{\l_t}
\newcommand{\xt}{\x_t}
\newcommand{\vt}{\v_t}
\newcommand{\et}{\e_t}

\newcommand{\bea}{\begin{eqnarray}}
\newcommand{\eea}{\end{eqnarray}}

\newcommand{\nn}{\nonumber}
\newcommand{\ds}{\displaystyle}

\newtheorem{theorem}{Theorem}[section]
\newtheorem{prop}[theorem]{Proposition}
\newtheorem{lemma}[theorem]{Lemma}
\newtheorem{claim}[theorem]{Claim}
\newtheorem{corollary}[theorem]{Corollary}
\newtheorem{fact}[theorem]{Fact}
\newtheorem{definition}[theorem]{Definition}
\newtheorem{remark}[theorem]{Remark}
\newtheorem{example}[theorem]{Example}
\newtheorem{sigmodel}[theorem]{Model}
\newtheorem{assu}[theorem]{Assumption}
\renewcommand\thetheorem{\arabic{section}.\arabic{theorem}}

\newcommand{\snr}{\text{SNR}}

\newcommand{\tk}{\tilde{k}}
\newcommand{\tl}{{\ell}}
\newcommand{\totl}{L}

\newcommand{\bfpara}[1]{ {\bf #1. }} 
\newcommand{\Item}{\item} 

\newcommand{\SEF}{\SE_F}

\renewcommand{\P}{\bm{P}}
	\newcommand{\kron}{\otimes}
\newcommand{\Uvec}{{\U_{vec}}}
\newcommand{\Zvec}{{\Z_{vec}}}
\newcommand{\ym}{{\y_{(mag)}}}
\renewcommand{\forall}{\text{ for all }}

\newcommand{\eps}{\epsilon}
\newcommand{\ev}{\mathcal{E}}

\newcommand{\Section}[1]{\section{#1}} 
\newcommand{\Subsection}[1]{\subsection{#1}} 

\renewcommand{\bhat}{\b}  \renewcommand{\Bhat}{\B}
\renewcommand{\xhat}{\x}  \renewcommand{\Xhat}{\X}

\newcommand{\deltaFt}{\delta_{F,t}}
\newcommand{\deltaFzero}{\delta_{F,0}}

\newcommand{\deltaF}{{\delta^F}}

\newcommand{\sigmamin}{\sigma_{\min}} \newcommand{\sigmamax}{\sigma_{\max}}

\newcommand{\inperr}{\mathrm{InpErr}}

\newcommand{\epsconscalar}{\eps_{\con,sc}}

\newcommand{\trnc}{\mathrm{trnc}}
\newcommand{\inp}{\mathrm{in}}

\newcommand{\errph}{ \ \mathrm{ErrPh} \ }

\newcommand{\mx}{\mathrm{mx}}

\title{Fast Low Rank column-wise Sensing: short proof}
\title{A Simple Proof for Fast and Sample-Efficient Federated Low Rank Matrix Recovery from column-wise Linear Projections}
\title{A Simple Proof for \\ Efficient Federated Low Rank Matrix Recovery from column-wise Linear Projections}
\title{Efficient Federated Low Rank Matrix Recovery via Alternating GD and Minimization: A Simple Proof}
\author{Namrata Vaswani \\
Iowa State University, Ames, IA, USA. \\
Email: namrata@iastate.edu
}
\maketitle

\begin{abstract}
This note provides a significantly simpler and shorter proof of our sample complexity guarantee for solving the low rank column-wise sensing problem using the Alternating Gradient Descent (GD) and Minimization (AltGDmin) algorithm. AltGDmin was developed and analyzed for solving this problem in our recent work. We also provide an improved guarantee.
\end{abstract}

\section{Introduction}
We study the low rank column-wise sensing (LRCS) problem which involves recovering a low rank matrix from independent compressive measurements of each of its columns. This problem occurs in dynamic MRI \cite{lrpr_gdmin_mri_jp} and in multi-task linear representation learning for few-shot learning \cite{du2020few}.
The alternating gradient descent (GD) and minimization (AltGDmin) algorithm for solving it in a fast, communication-efficient and private fashion was developed and analyzed in our recent work \cite{lrpr_gdmin}. 
%
This short paper provides a significantly simpler and shorter proof of our sample complexity guarantee for AltGDmin. In fact, it also improves the sample complexity needed by the AltGDmin iterations by a factor of $r$. 

\section{Problem statement, notation, and algorithm}
\subsection{Problem statement and assumption and notation}
The goal is to recover an $n \times q$ rank-$r$ matrix $\Xstar =[\xstar_1, \xstar_2, \dots, \xstar_q]$ from  $m$ linear projections (sketches) of each of its $q$ columns, i.e. from
\bea
\y_k := \A_k \xstar_k, \ k  \in [q]
\label{ykvec}
\eea
where each $\y_k$ is an $m$-length vector,  $[q]:=\{1,2,\dots, q\}$, and the measurement/sketching matrices $\A_k$ are mutually independent and known. The setting of interest is  low-rank (LR), $r \ll \min(n,q)$, and undersampled measurements, $m < n$. Each $\A_k$ is assumed to be random-Gaussian: each entry of it is independent and identically distributed (i.i.d.) standard Gaussian.
Let
	$
	\Xstar \svdeq \Ustar  \bSigma \Vstar: = \Ustar \Bstar
	$
	denote its reduced (rank $r$) SVD, and $\kappa:= \sigmax/\sigmin$ the condition number of $\bSigma$. We let $\Bstar:= \bSigma \Vstar$. 

Since no measurement $\y_\ik$ is a global function  of the entire matrix, $\Xstar$, we need the following assumption, borrowed from LR matrix completion literature, to make our problem well-posed (allow for correct interpolation across columns).
	
\begin{assu}[Incoherence of right singular vectors] \label{right_incoh}
Assume that $\|\bstar_k\|^2 \le \mu^2 r \sigmax^2 / q$ for a numerical constant $\mu$.
\end{assu}

In our discussion of communication complexity and privacy, we assume a vertically federated setting: different subsets of $\y_k, \A_k$ are available at different nodes.

\subsubsection{Notation}
We use $\|.\|_F$ to denote the Frobenius norm,  $\|.\|$ without a subscript for the (induced) $l_2$ norm, $^\top$ to denote matrix or vector transpose, $\e_k$ to denote the $k$-th canonical basis vector ($k$-th column of $\I$), and $\M^\dag := (\M^\top \M)^{-1} \M^\top$. 
For two $n \times r$ matrices $\U_1, \U_2$ that have orthonormal columns, we use
\[
\SE_2(\U_1, \U_2) : = \|(\I  - \U_1 \U_1^\top) \U_2\|
\]
as the Subspace Distance (SD) measure. 
In our previous work \cite{lrpr_gdmin}, we used the Frobenius norm SD,
\[
\SEF(\U_1, \U_2) : = \|(\I  - \U_1 \U_1^\top) \U_2\|_F.
\]
Clearly, $\SE_F(\U_1, \U_2) \le \sqrt{r} \SE_2(\U_1, \U_2)$.
We reuse the letters $c,C$ to denote different numerical constants in each use with the convention that $c < 1$ and $C \ge 1$. We use $\sum_k$ as a shortcut for the summation over $k=1$ to $q$ and $\sum_\ik$ for the summation over $i=1$ to $m$ and $k=1$ to $q$. {\em We use whp to refer to ``with high probability'' and this means that the claim holds with probability (w.p.) at least $ 1- n^{-10}$.}

\subsection{Review of AltGDmin algorithm \cite{lrpr_gdmin}}
AltGDmin, summarized in Algorithm \ref{gdmin}, imposes the LR constraint by factorizing the unknown matrix $\X$ as $\X = \U \B$ with $\U$ being an $n \times r$ matrix and $\B$ an $r \times q$ matrix. It minimizes $
f(\U,\B): = \sum_{k=1}^q \|\y_k - \U \b_k\|^2
$ as follows:%
\ben
\item {\em Truncated spectral initialization:} Initialize $\U$ (see below).
\item At each iteration, update $\B$ and $\U$ as follows:
\ben
\item {\em Minimization for $\B$:} keeping $\U$ fixed, update $\B$ by solving $\min_{\B} f(\U, \B)$. Due to the form of the LRCS model, this minimization decouples across columns, making it a cheap least squares problem of recovering $q$ different $r$ length vectors. It is solved as $\b_k = (\A_k \U)^\dag \y_k$ for each $k \in [q]$.%

\item {\em Projected-GD for $\U$:} keeping $\B$ fixed, update $\U$ by a GD step, followed by orthonormalizing its columns: 
    $\U^+  = QR(\U - \eta \nabla_{\U} f(\U,\B)))$. Here $QR(.)$ orthonormalizes the columns of its input.
\een
\een
We initialize $\U$ by computing the top $r$ singular vectors of
\[
\X_0 := \sum_k \A_k^\top \y_{k,trnc} \e_k^\top, \ \y_{k,trnc}:= \mathrm{trunc}(\y_k,\alpha) 
\]
Here $ \alpha:= \tilde{C} \sum_k \|\y_k\|^2 / mq$ with $\tilde{C}:=9\kappa^2\mu^2$. The function $\mathrm{trunc}$ truncates (zeroes out) all entries of $\y_k$ with magnitude greater than $\sqrt\alpha$, i.e.,  for all $j\in [n]$, $\mathrm{trunc}(\y,\alpha)_j = (\y)_j  \indic_{|\y_j| \le  \sqrt\alpha})$, with $\indic$ being the indicator function.

Sample-splitting is assumed, i.e., each new update of $\U$ and $\B$ uses a new independent set of measurements and measurement matrices, $\y_k, \A_k$.

 The use of minimization to update $\B$ at each iteration is what helps ensure that we can show exponential error decay with a constant step size. At the same time, due to the column-wise decoupled nature of LRCS, the time complexity for this step is only as much as that of computing one gradient w.r.t. $\U$. Both steps need time\footnote{The LS step time is $\max(q \cdot mnr, q \cdot mr^2)=mqnr$ (maximum of the time  needed for computing $\A_k \U$ for all $k$, and that for obtaining $\b_k$ for all $k$)  while the GD step time is $\max( q \cdot mnr, nr^2) =mqnr$ (maximum of the time needed for computing the gradient w.r.t. $\U$, and time for the QR step).}of order $mqnr$. This is only $r$ times more than ``linear time'' (time needed to read the algorithm inputs, here $\y_k,\A_k$'s). To our knowledge, $r$-times linear-time is the best known time complexity for any algorithm for any LR matrix recovery problem. Moreover, due to the use of the $\X=\U\B$ factorization, AltGDmin is also communication-efficient. Each node needs to only send $nr$ scalars (gradients w.r.t $\U$) at each iteration.%

\begin{algorithm}[t]
\caption{\small{The AltGD-Min algorithm.   
}}
\label{gdmin}
\begin{algorithmic}[1]
   \State {\bfseries Input:} $\y_k, \A_k, k \in [q]$

\State {\bfseries Sample-split:} Partition the measurements and measurement matrices into $2T+2$ equal-sized disjoint sets: two sets for initialization and $2T$ sets for the iterations. Denote these by $\y_k^{(\tau)}, \A_k^{(\tau)}, \tau=00,0,1,\dots 2T$.

  \State {\bfseries Initialization:}
\State Using $\y_k \equiv \y_k^{(00)}, \A_k \equiv \A_k^{(00)}$,
\State set $\alpha \leftarrow \tC \frac{1}{mq}\sum_\ik\big|\y_\ik\big|^2$,

\State Using $\y_k \equiv \y_k^{(0)}, \A_k \equiv \A_k^{(0)}$,
\State  set $\y_{k,trunc}(\alpha) \leftarrow \y_{k,trnc}:= \mathrm{trunc}(\y_k,\alpha)$, 
\State  set $\ds  \Xhat_{0} \leftarrow (1/m) \sum_{k \in [q]} \A_k^\top \y_{k,trunc}(\alpha) \e_k^\top$

%
\State   set $\U_0 \leftarrow $ top-$r$-singular-vectors of $\Xhat_0$
\State {\bfseries GDmin iterations:}

   \For{$t=1$ {\bfseries to} $T$}

   \State  Let $\U \leftarrow \U_{t-1}$.
\State Using $\y_k \equiv \y_k^{(t)}, \A_k \equiv \A_k^{(t)}$,
\State  set $\bhat_k   \leftarrow  (\A_k \U)^\dagger \y_k$, $\xhat_k \leftarrow \U \bhat_k$ for all $k \in [q]$

\State 
Using $\y_k \equiv \y_k^{(T+t)}, \A_k \equiv  \A_k^{(T+t)}$, compute
\State set $\nabla_\U f(\U, \B) =  \sum_k \A_k^\top (\A_k \U \bhat_k - \y_k) \bhat_k ^\top$
%
\State  
set $\ds \Uhat^+   \leftarrow \U - (\eta/m) \nabla_\U f(\U, \B_t)$.

 \State 
 compute $\Uhat^+ \qreq \U^+ \R^+$.
 \State Set $\U_t \leftarrow \U^+$.
     \EndFor
\end{algorithmic}
\end{algorithm}

\section{New Guarantee}
Let $m_0$ denote the total number of samples per column needed for initialization and let $m_1$ denote this number for each GDmin iteration. Then, the total sample complexity per column is $m = m_0  + m_1 T$.
Our guarantee given next provides the required minimum value of $m$.
\begin{theorem}
Assume that Assumption \ref{right_incoh} holds. Set $\eta = 0.4 / m \sigmax^2 $ and $T = C \kappa^2 \log( 1 /\eps)$.
If
\[
m q \ge C \kappa^4 \mu^2 (n+q) r ( \kappa^4 r + \log (1/\eps) ) 
\]
and $m \ge C \max(\log n, \log q, r) \log (1/\eps)$, then, with probability (w.p.) at least $1 - n^{-10}$,
\[
\SE_2(\U,\Ustar) \le \eps \text{ and } \|\x_k - \xstar_k\| \le \eps \|\xstar_k\| \ \text{for all $k \in [q]$}. 
\]
%
The time complexity is $mqnr \cdot T =   mqnr \cdot \kappa^2 \log( 1 /\eps)$. The communication complexity is $nr\cdot T = nr  \cdot \kappa^2 \log( 1 /\eps)$ per node.
\label{main_res}
\end{theorem}

\begin{proof} 
We prove the three results needed for proving this in Sections \ref{proofs_min}, \ref{proofs_gd}, and \ref{proofs_init} below. We use these to prove the above result in Appendix \ref{proof_main_res}.%
\end{proof}

\subsection{Discussion}
We use $a \gtrsim b$ to mean that $a \ge C_{\kappa,\mu} b$ where $C_{\kappa,\mu}$ includes terms dependent on $\kappa,\mu$.
Most of our discussion treats $\kappa,\mu$ as numerical constants that do not grow with $n,q,r$.
\cblue
This assumption is borrowed from the rich past literature on various other LR matrix recovery problems, e.g., see \cite{lowrank_altmin,rpca_gd}. 
\cbl
Also, below, whp means w.p. at least $ 1- n^{-10}$.


\cblue
In this work (as well as in other works that use matrix factorization to solve LR recovery problems), the goal is to obtain a bound of the form $\SE_2(\U,\Ustar) \le \eps$.  If a bound on $\SE_F$ is needed, one can use $\SE_F(\U,\Ustar) \le \sqrt{r}\SE_2(\U,\Ustar) \le \sqrt{r} \eps$. Using Lemma \ref{bhat_lemma} given below and triangle inequality, the bound on $\SE_2$ helps get the bound $\|\x_k - \xstar_k\| \lesssim \eps \|\xstar_k\|$.

Our older result from \cite{lrpr_gdmin} used $\SE_F(\U,\Ustar)$ in its analysis and showed that, to guarantee $\SE_F(\U,\Ustar) \le \eps$ whp, we need
\[
m q \ge C \kappa^4 \mu^2 (n+q) r^2 (\kappa^4+ \log (1/\eps) )  
\]
Since $\SE_2(\U,\Ustar) \le \SE_F(\U,\Ustar)$, this implies that we need this same complexity also to guarantee $\SE_2(\U,\Ustar) \le \eps$.
%
Our new result needs $m q \ge C \kappa^4 \mu^2 (n+q) r ( \kappa^4 r + \log (1/\eps) )$ to guarantee $\SE_2(\U,\Ustar) \le \eps$. Thus, this new result improves the dependence on $r,\eps$ from order $r^2 \log(1/\eps)$ to order $r \max(r,\log(1/\eps))$. This is an improvement over the old result by a factor of $\min(r,\log(1/\eps))$.
%
%
%
This improvement is obtained because our new guarantee for the GD step (Theorem \ref{gd_claim}) only needs $m q \gtrsim nr$ at each iteration. On the other hand, the older result needed $m q \gtrsim nr^2$. Both guarantees need $m q \gtrsim nr^2$ for initialization, see Theorem \ref{init_claim}.%

 \color{black}


We are able to improve our result because we now use a simpler proof technique that works for the LRCS problem, but not for its LR phase retrieval (LRPR) generalization. LRPR involves recovering $\Xstar$ from $\z_k:= |\y_k|, k \in [q]$. In \cite{lrpr_gdmin}, we were attempting to solve both problems.
There are two differences in our new proof compared with that of \cite{lrpr_gdmin}: (i) we use the 2-norm subspace distance $\SE_2(\U,\Ustar)$ instead of $\SE_F(\U,\Ustar)$, and (ii) we do not use the fundamental theorem of calculus  \cite{lan93,pr_mc_reuse_meas} for analyzing the GD step, but instead use a much simpler direct approach. If we only use (ii) but not (i), we will still get a simpler proof, but we will not get the sample complexity gain. 
In \cite{lrpr_gdmin}, we used the Frobenius norm SD because it helped obtain the desired $nr^3$ guarantee for LRPR\footnote{Our older guarantee for LRPR \cite{lrpr_gdmin} needed $mq \gtrsim nr^2 (r + \log(1/\eps))$. If we use the new approach developed here, it will need  $mq \gtrsim nr (r^3 + \log(1/\eps))$.}. 
%
Also, in hindsight, the use of the fundamental theorem of calculus was unnecessary. It has been used in earlier work \cite{pr_mc_reuse_meas} for analyzing a GD based algorithm for standard PR  and for LR matrix completion and that originally motivated us to adapt the same approach for LRCS and LRPR.

\cbl 
Both the current result and our old one have the same dependence on $\kappa$. After initialization, the sample complexity grows roughly as $\kappa^4$. A similar dependence on $\kappa$ also exists in all past sample complexity guarantees for iterative -- AlMin or GD -- algorithms for LR matrix sensing, LR matrix completion or robust PCA, e.g., see \cite{lowrank_altmin,rpca_gd}; and also for our older work on  AltMin for a generalization of LRCS, LR phase retrieval \cite{lrpr_best}.
One way to improve this dependence is to develop a stage-wise algorithm similar to that introduced in \cite{lowrank_altmin,rmc_gd} for LR matrix sensing and completion respectively. Developing this for the current LRCS problem is an open future work question. A second option is to use the convex relaxation approaches which usually depend on lower powers of $\kappa$. However, these need a much higher iteration complexity, making them much slower.



 \color{black}

\section{New proof: GDmin iterations} \label{proofs}

\subsection{Definitions and preliminaries}
Let $\U$ be the estimate at the $t$-th iteration. Define
\begin{align*}
\g_k &  := \U^\top \xstar_k, k \in [q], \text{ and } \G:= \U^\top \Xstar, \\
\P &  := \I - \Ustar \Ustar^\top, \\
\gradU &  :=  \nabla_\U f(\U,\Bhat) = \sum_k \A_k^\top (\A_k \U \bhat_k - \y_k) \bhat_k ^\top  \\
& =\sum_\ik (\y_\ik - \a_\ik{}^\top \U \bhat_k) \a_\ik \bhat_k{}^\top
\end{align*}

For an $n_1 \times n_2$ matrix, $\Z$, $\sigmamin(\Z)=\sigmamin(\Z^\top) =\sigma_{\min(n_1,n_2)}(\Z)$. Thus, if $\A$ is tall, then $\sigmamin(\A)= \sqrt{\lambda_{\min}(\A^\top \A)}$. Using this, it follows that, if $\A = \B \C$ and $\A$ and $\B$ are tall (or square), then  $\sigmamin(\A) \ge \sigmamin(\B) \sigmamin(\C)$.


\subsection{Minimization step}\label{proofs_min}

Assume $\SE_2(\U,\Ustar) \le \delta_t$ with $\delta_t < 0.02$.

We use the following lemma from \cite{lrpr_gdmin}.
\begin{lemma}[\cite{lrpr_gdmin}]
\label{bhat_lemma}
Let 
$\g_k: = \U^\top \xstar_k.$
Then, w.p. at least $1 -  \exp(\log q + r - c \eps m)$, for all $k \in [q]$,
\begin{align}
\|\g_k - \bhat_k \|   & \leq  1.2 \eps \|\left(\I_n-\U\U^\top \right)\Ustar \bstar_k\| \nonumber
\end{align}
\end{lemma}
\begin{proof}
We provide a proof of this lemma in Appendix \ref{proof_bhat_lemma} to emphasise a slightly more general version of this lemma. In particular, our proof shows that we can replace $0.4$ in the previous version of this lemma by any $\eps > 0$, and the bound holds w.p. at least $1 -  \exp(\log q + r - c \eps m)$.
\end{proof}

By Lemma \ref{bhat_lemma} with $\eps=0.3$, if $m \gtrsim \max(\log q, \log n, r)$, then, whp, for all $k \in [q]$,
\[
\|\b_k - \g_k\| \le 0.4 \|(\I - \U \U^\top) \Ustar \bstar_k\|
\]
Using $\SE_2(\U,\Ustar) \le \delta_t$, this directly implies that
\ben
\item $\|\b_k - \g_k\| \le 0.4 \delta_t \|\bstar_k\|$

\item $\|\b_k\| \le \|\g_k\|+0.4 \cdot 0.02 \|\bstar_k\| \le 1.1 \|\bstar_k\|$

\item $\|\x_k - \xstar_k\| \le 1.4  \delta_t \|\bstar_k\|$
\een
(to obtain the last bound, we need to add and subtract $\U \g_k$ and then use triangle inequality).
Using these bounds,
\[
\|\B - \G\|_F \le 0.4 \delta_t \sqrt{\sum_k \|\bstar_k\|^2} = 0.4 \delta_t \sqrt{r} \sigmax
\]
and we can get a similar bound on $\|\X - \Xstar\|_F$. Thus,
\ben
\item $\|\B - \G\|_F \le 0.4 \delta_t \|\Bstar\|_F \le  0.4 \sqrt{r} \delta_t \sigmax $

\item $\|\X - \Xstar\|_F \le 1.4 \sqrt{r} \delta_t \sigmax $

\een
Furthermore,
\[
\sigmamin(\B) \ge \sigmamin(\G) - \|\B - \G\| \ge \sigmamin(\G) - \|\B - \G\|_F
\]
To lower bound $\sigmamin(\G)$, observe that
\[
\sigmamin(\G) = \sigmamin(\G^\top) \ge \sigmin \sigmamin(\Ustar^\top \U)
\]
and
\[
\sigmamin(\Ustar^\top \U)  = \sqrt{1 - \|\P \U\|^2} \ge \sqrt{1- \delta_t^2}.
\]
This follows using $\sigmamin^2(\Ustar^\top \U) = \lambda_{\min}( \U^\top \Ustar \Ustar^\top \U )=  \lambda_{\min}( \U^\top(\I -\P) \U) =   \lambda_{\min}(\I - \U^\top \P \U) =  \lambda_{\min}(\I - \U^\top \P^2 \U) = 1 -  \lambda_{\max}(\U^\top \P^2 \U) = 1 - \|\P \U\|^2$. 

Combining the above three bounds and using the bound on  $\|\B - \G\|_F$, if $\delta_t < 0.02 / \sqrt{r} \kappa$, then
\[
\sigmamin(\B) \ge \sqrt{1- \delta_t^2}\sigmin - 0.4 \sqrt{r} \delta_t \sigmax \ge 0.9 \sigmin
\]
and
\[
\sigma_{\max}(\B) \le  \|\G\| + 0.4 \sqrt{r} \delta_t \sigmax \le 1.1 \sigmax
\]
since $\|\G\| \le \|\Bstar\| = \sigmax$.
Thus, we have proved the following claim:
\begin{theorem}
Assume that $\SE_2(\U,\Ustar) \le \delta_t$.
If $\delta_t \le 0.02 / \sqrt{r} \kappa$,
and
if $m \gtrsim \max(\log q, \log n, r)$, then whp,
\ben
\item $\|\b_k - \g_k\| \le 0.4 \delta_t \|\bstar_k\|$

\item $\|\b_k\| \le \|\g_k\|+0.4 \cdot 0.02 \|\bstar_k\| \le 1.1 \|\bstar_k\|$

\item $\|\B - \G\|_F \le 0.4 \delta_t \|\Bstar\|_F \le  0.4 \sqrt{r} \delta_t \sigmax $

\item $\|\x_k - \xstar_k\| \le 1.4  \delta_t \|\bstar_k\|$

\item $\|\X - \Xstar\|_F \le 1.4 \sqrt{r} \delta_t \sigmax $

\item $\sigmamin(\B) \ge  0.9 \sigmin$

\item $\sigma_{\max}(\B) \le  1.1 \sigmax$
\een
(only the last two bounds require the upper bound on $\delta_t$).
\label{Blemma}
\end{theorem}

\subsection{New bounds on the expected gradient and deviation from it}\label{proof_grad_bnd_new}

Using independence of $\A_k$ and $\{\U,\b_k\}$ (due to sample splitting),
\[
\E[\gradU] = \sum_k m (\x_k-{\xstar_k})\b_k{}^\top
\]

Using bounds on $\|\B\|$ and  $\|\Xstar - \X\|_F$ from  Theorem \ref{Blemma}, if $\delta_t < \frac{c}{\sqrt{r} \kappa}$,
\begin{align}
\| \E[\gradU]\|   & =  \| \sum_k m (\x_k-{\xstar_k})\b_k{}^\top \| = m\|(\X - \Xstar) \B^\top \|   \nn\\
&\le m\|\X - \Xstar\| \cdot \| \B\|  \nn\\
&  \le m\| \X -\Xstar\|_F \cdot \| \B\|  \le 1.1 m\delta_t \sqrt{r}  \sigmax^2. \nn
\end{align}
w.p. $1- \exp(\log q+r-cm)$.

Next, we bound
$\| \gradU-\E[\gradU]\|
 = \max_{\|\w\|=1, \|\z\|=1}  \w^\top ( \sum_k \sum_i \a_\ik \a_\ik^{\top}(\x_k-\xstar_k) \b_k^\top  - \E[\cdot]) \z.
 $
 This also uses independence of $\A_k$ and $\{\U,\b_k\}$.

We bound the above for fixed unit norm $\w,\z$ using sub-exponential Bernstein inequality (Theorem 2.8.1 of \cite{versh_book}). We extend the bound to all unit norm $\w,\z$ by using a standard epsilon-net argument.
For fixed unit norm $\w,\z$, consider
\[
\sum_k \sum_i \left( (\w^\top \a_\ik)  (\b_k^\top \z) \a_\ik^{\top}(\x_k-\xstar_k)- \E[\cdot] \right)
\]
Observe that the summands are independent, zero mean, sub-exponential r.v.s with sub-exponential norm $K_\ik \le C \| \w \| \cdot |\z^\top \b_k| \cdot \| \x_k-\xstar_k\| =C |\z^\top \b_k| \cdot \| \x_k-\xstar_k\| $.
We apply the sub-exponential Bernstein inequality, Theorem 2.8.1 of \cite{versh_book}, with $t = \eps_1 \delta_t m  \sigmin^2$. We have
\begin{align*}
\frac{t^2}{\sum_\ik K_\ik^2}  
& \ge \frac{\eps_1^2\delta_t^2 m^2  \sigmin^4}{ m  \sum_k \|\x_k-\xstar_k\|^2 (\z^\top \b_k)^2   }\\
& \ge \frac{\eps_1^2\delta_t^2 m^2  \sigmin^4}{ m \max_k \|\x_k-\xstar_k\|^2 \sum_{k} (\z^\top \b_k)^2 }\\
& = \frac{\eps_1^2\delta_t^2 m^2  \sigmin^4}{ m \max_k \|\x_k-\xstar_k\|^2 \|\z^\top \B\|^2 }\\
& \ge \frac{\eps_1^2\delta_t^2 m^2  \sigmin^4 q }{ 1.4^2 m \mu^2 \delta_t^2 r \sigmax^2 \|\B\|^2 }\\
& \ge \frac{\eps_1^2\delta_t^2 m^2  \sigmin^4 q }{ 1.4^2 m \mu^2 \delta_t^2 r \sigmax^2 1.1 \sigmax^2 } =  c \frac{\eps_1^2 m  q }{ \kappa^4 \mu^2  r }\\
\frac{t}{\max_\ik K_\ik} & \ge  \frac{\eps_1\delta_t m  \sigmin^2 }{\max_k \| \b_k\| \max_k \|\x_k - \xstar_k\| }
\ge c \frac{\eps_1 m  q }{\kappa^2 \mu^2  r}
\end{align*}
In the above, we used
(i) $\sum_k (\z^\top \b_k)^2 = \|\z^\top \B\|^2 \le \|\B\|^2$  since $\z$ is unit norm,
(ii) Theorem \ref{Blemma} to bound $ \|\B\| \le 1.1 \sigmax$, and
(iii) Theorem \ref{Blemma} followed by Assumption \ref{right_incoh} (right incoherence) to bound $\|\x_k-\xstar_k\| \le \delta_t \cdot \mu \sigmax \sqrt{r/q}$ and $|\z^\top \b_k| \le \|\b_k\| \le 1.1 \|\bstar_k\| \le 1.1 \mu \sigmax \sqrt{r/q}$.

For $\eps_1 < 1$, the first term above is smaller (since $1/\kappa^4 \le 1/\kappa^2$), i.e., $\min(\frac{t^2}{\sum_\ik K_\ik^2},\frac{t}{\max_\ik K_\ik} ) =  c\frac{\eps_1^2 m q }{\kappa^4 \mu^2 r}.$
Thus, by sub-exponential Bernstein,  w.p. at least  $1-\exp (- c\frac{\eps_1^2 m q }{\kappa^4 \mu^2 r} )- \exp(\log q+r-cm)$, for a given $\w,\z$,
\[
	 \w^\top(\gradU - \E[\gradU] ) \z \le  \eps_1 \delta_t m \sigmin^2
\]
Using a standard epsilon-net argument to bound the maximum of the above over all unit norm $\w,\z$, e.g., using \cite[Proposition 4.7]{lrpr_gdmin}, we can conclude that
\[
	 \|\gradU - \E[\gradU] \| \le  1.1 \eps_1 \delta_t m \sigmin^2
\]
w.p. at least  $1-  \exp ( C(n+r) -c\frac{\eps_1^2 m q }{\kappa^4 \mu^2 r} )- \exp(\log q+r-cm)$. The factor of $ \exp(C(n+r))$ is due to the epsilon-net over $\w$ and that over $\z$: $\w$ is an $n$-length unit norm vector while $\z$ is an $r$-length unit norm vector. The smallest epsilon net covering the hyper-sphere of all $\w$s  is of size $(1+ 2/\eps_{net})^n = C^n$ with $\eps_{net}=c$ while that for $\z$ is of size $C^r$. Union bounding over both thus gives a factor of $C^{n+r}$.
By replacing $\eps_1$ by $\eps_1/1.1$, our bound becomes simpler (and $1/1.1^2$ gets incorporated into the factor $c$). We have thus proved the following.

\begin{lemma}
\label{grad_bnd_new}
Assume that $\SE_2(\U,\Ustar) \le \delta_t$. The following hold:
\ben
\item $\E[\gradU] = m (\X - \Xstar) \B^\top = m (\U \B \B^\top - \Xstar \B^\top)$
\item $\|\E[\gradU]\| \le 1.1 m\delta_t \sqrt{r}  \sigmax^2$
\item If $\delta_t < \frac{c}{\sqrt{r} \kappa}$, then, w.p. at least
\\
$1-  \exp ( C(n+r) -c\frac{\eps_1^2 m q }{\kappa^4 \mu^2 r} )- \exp(\log q+r-cm)$, 
 \[
\|\gradU - \E[\gradU]\| \le  \eps_1 \delta_t m \sigmin^2
\]
\een
\end{lemma}
The above lemma is an improvement over the bounds given in \cite{lrpr_gdmin} because $\delta_t$ is now the bound on the 2-norm SD, and still it only needs $mq \gtrsim nr/\eps_1^2$.

\subsection{GD step}\label{proofs_gd}
Assume that $\SE_2(\U,\Ustar) \le \delta_t$ with $\delta_t < 0.02$.


Recall the Projected GD step for $\U$:
\begin{eqnarray}
\Utilde^+ &= \U - \eta \gradU \text{ and } \Utilde^+  \qreq \U^+\R^+  \nn
\end{eqnarray}

Since $ \U^+= \Utilde^+ (\R^+)^{-1}$ and since $\|(\R^+)^{-1}\| =1/ \sigmamin(\R^+) = 1/ \sigmamin(\Utilde^+)$, thus, $\SE_2(\U^+, \Ustar)  = \| \P \U^+\| $ can be bounded as
\begin{eqnarray}
\SE_2(\U^+, \Ustar) & \le \dfrac{\|\P \Utilde^+\|}{\sigmamin(\Utilde^+)} \le \dfrac{\|\P \Utilde^+\|}{\sigmamin(\U) - \eta \|\gradU\|}
\label{SDeq}
\end{eqnarray}

Consider the numerator.
Adding/subtracting $\E[\gradU]$, left multiplying both sides by $\P$, and using Lemma \ref{grad_bnd_new} (first part),
\begin{align*}
\tilde\U^+ & = \U - \eta  \E[\gradU] + \eta  (\E[\gradU] - \gradU ), \text{ thus,} \\
 \P \tilde\U^+  
 & = \P \U - \eta m \P \U \B \B^\top + \eta \P (\E[\gradU] - \gradU)
\end{align*}
The last row used $\P \Xstar=0$. 
Thus,
\begin{align}
\|\P \Utilde^+ \| \le \|\P \U\| \|\I -  \eta m \B \B^\top\| + \eta \|\E[\gradU] - \gradU\| \label{eq:PU+}
\end{align}
Using Theorem \ref{Blemma}, we get
\[
\lambda_{\min}(\I -  \eta m \B \B^\top) = 1 - \eta m \|\B\|^2 \ge 1 - 1.2 \eta m \sigmax^2
\]
Thus, if $\eta < 0.5/ m \sigmax^2 $, then the above matrix is p.s.d. This along with Theorem~\ref{Blemma} then implies that
\[
\|\I -  \eta m \B \B^\top\| = \lambda_{\max}(\I -  \eta m \B \B^\top) \le 1 - 0.9 \eta m  \sigmin^2
\]
Using the above, \eqref{eq:PU+}, and the bound on $\|\E[\gradU] - \gradU\|$ from Lemma \ref{grad_bnd_new}, we conclude the following:
If $\eta \le 0.5/m\sigmax^2$, and $\delta_t \le c/\sqrt{r} \kappa$, then
\begin{align}
\|\P \tilde\U^+ \|
& \le \|\P \U\| \|\I -  \eta m \B \B^\top\| + \eta \|\E[\gradU] - \gradU\| \nn \\
& \le \delta_t (1 - 0.9~\eta m \sigmin^2) + \eta  m \eps_1 \delta_t \sigmin^2
\label{PUtilde}
\end{align}
w.p. at least $1-  \exp ( C(n+r) -c\frac{\eps_1^2 m q }{\kappa^4 \mu^2 r} )- \exp(\log q+r-cm)$. This probability is at least $1 - n^{-10}$ if $mq \gtrsim \kappa^4 \mu^2 nr/\eps_1^2$ and $m \gtrsim \max(\log n, \log q, r)$.

Next we use \eqref{PUtilde} with $\eps_1 = 0.1$  and Lemma \ref{grad_bnd_new} to bound the right hand side of \eqref{SDeq}. Set $\eta = c_\eta / m \sigmax^2$. 
If  $c_\eta \le 0.5$, if $\delta_t \le c/\sqrt{r} \kappa^2$, and lower bounds on $m$ from above hold, \eqref{SDeq} implies that, whp,
\begin{align}
& \SE_2(\U^+, \Ustar) \nn \\
& \le \dfrac{\|\P \Utilde^+\|}{\sigmamin(\Utilde^+)} \nn \\
& \le \dfrac{\|\P \Utilde^+\|}{\sigmamin(\U) - \eta \|\gradU\|} \nn \\
& \le \dfrac{\|\P \U\| \|\I -  \eta m \B \B^\top\| + \eta \|\E[\gradU] - \gradU\|}{1- \eta \|\E[\gradU]\| - \eta \|\gradU - \E[\gradU]} \nn \\
& \le \frac{\delta_t  (1 - \eta m \sigmin^2 (0.9  - 0.1) ) }{1 - \eta \|\E[\gradU]\| - \eta \|\gradU - \E[\gradU] \|} \nn \\
& \le \frac{\delta_t (1 - 0.8 \eta m  \sigmin^2)}{1 - \eta m \delta_t \sqrt{r} \sigmax^2 (1.4    + \frac{0.1}{\kappa^2 \sqrt{r}} )  } \nn  \\
& \le \frac{\delta_t (1 - 0.8  \eta m \sigmin^2)}{1 -  1.5 \eta m \delta_t \sqrt{r} \sigmax^2 } \nn \\
& \le \delta_t (1 - 0.8  \eta m \sigmin^2) (1 +  1.5 \eta m \delta_t \sqrt{r} \sigmax^2 ) \nn  \\
& \le \delta_t (1 - 0.8  \eta m  \sigmin^2 +  1.5 \eta m \delta_t \sqrt{r} \sigmax^2 )\nn   \\
& = \delta_t (1 - \eta m  \sigmin^2 (0.8 -  1.5 \delta_t \sqrt{r} \kappa^2 ) ) \nn  \\
& \le \delta_t (1 - \eta m  \sigmin^2 (0.8 -  0.15 ) ) \nn \\
& \le \delta_t (1 - 0.6 \eta m  \sigmax^2 / \kappa^2 ) = \delta_t (1 - 0.6 c_\eta / \kappa^2 ) \nn
\end{align}
%
In the above we used $\kappa^2\sqrt{r} >1$, $(1-x)^{-1} < (1+x)$ if $|x| < 1$, and  $\delta_t < 0.1 / \sqrt{r} \kappa^2$ (used in second-last inequality).

Thus, we have proved the following result.
\begin{theorem}\label{gd_claim}
Assume that Assumption \ref{right_incoh} holds and $\SE_2(\U,\Ustar) \le \delta_t$.
If $\delta_t \le 0.02 / \sqrt{r} \kappa^2$,
if $\eta = c_\eta / m \sigmax^2 $ with $c_\eta \le 0.5$,
and
if $mq \gtrsim \kappa^4 \mu^2 n r$ and $m \gtrsim \max(\log n, \log q, r)$, then, whp
\[
\SE_2(\U^+, \Ustar) \le \delta_{t+1}:= \delta_t (1 - 0.6 c_\eta / \kappa^2 )
\]
\end{theorem}

\section{New proof: Initialization}\label{proofs_init}
We need a new result for the initialization step because we need a tight bound on $\SE_2(\U_0, \Ustar)$.

\subsection{Results taken from \cite{lrpr_gdmin}}
Recall from Algorithm \ref{gdmin} that $\alpha$ uses a different set of measurements that is independent of those used for $\X_0$. We use the following four results from \cite{lrpr_gdmin}.
\begin{lemma}[\cite{lrpr_gdmin}] \label{EX0}
Conditioned on $\alpha$, we have the following conclusions.
Let $\zeta$  be a scalar standard Gaussian r.v.. Define
\[
\beta_{k}(\alpha) := \E[\zeta^2 \indic_{\{\|\xstar_{k}\|^2\zeta^2 \leq  \alpha\}}].
\]
Then,
\begin{align*}
&\E[\Xhat_0|\alpha] = \Xstar \D(\alpha), \\
&\text{ where }  \D(\alpha):=diagonal(\beta_k(\alpha),k \in [q])
\end{align*}
i.e. $\D(\alpha)$ is a diagonal matrix of size $q\times q$ with diagonal entries $\beta_{k}$ defined above.
\end{lemma}


\begin{fact}[\cite{lrpr_gdmin}] \label{sumyik_bnd}
Let
\[
\ev:= \left\{ \tC(1 - \epsilon_1) \frac{\|\Xstar\|_F^2}{q} \le \alpha \le \tC (1 + \epsilon_1) \frac{\|\Xstar\|_F^2}{q}  \right\}.
\]
$\Pr(\alpha \in \ev) \ge 1 - \exp(- \tc mq \epsilon_1^2)$.  Here  $\tc = c/\tC = c / \kappa^2 \mu^2.$
\end{fact}

\begin{fact}[\cite{lrpr_gdmin}]
\label{betak_bnd}
		For any $\epsilon_1 \leq 0.1$, 
		$
 \min_k  \E\left[\zeta^2 \indic_{ \left\{ |\zeta| \leq \tC \frac{ \sqrt{1- \epsilon_1} \|\Xstar\|_F }{ \sqrt{q}\|\xstar_{k}\| } \right\} } \right] \geq 0.92.
	$
	\end{fact}

\begin{lemma}[\cite{lrpr_gdmin}]  \label{init_bnd}
Fix $0 < \epsilon_1 < 1$. Then,  w.p. at least $1-\exp\left[(n+q)-c\epsilon_1^2mq/\mu^2\kappa^2\right]$, conditioned on $\alpha$, for an $\alpha \in \ev$,
		\[
		\|\Xhat_{0} -\E[\Xhat_{0}|\alpha]\| \leq 1.1 \epsilon_1 \|\Xstar\|_F
		\]
\end{lemma}

Facts \ref{sumyik_bnd} and \ref{betak_bnd} together imply that, w.p. at least $1 - \exp(- \tc mq \epsilon_1^2)$,
\bea
\min_k \beta_k(\alpha) \ge \min_k  \E\left[\zeta^2 \indic_{\{ |\zeta| \leq \tC \frac{ \sqrt{1- \epsilon_1} \|\Xstar\|_F }{ \sqrt{q}\|\xstar_{k}\| } \} } \right]  \ge 0.92.
\label{minbetak}
\eea
The first inequality is an immediate consequence of Fact \ref{sumyik_bnd} ($\alpha \in \ev$) and the second follows by Fact \ref{betak_bnd}.

By setting $\eps_1 = \eps_0/ 1.1 \sqrt{r} \kappa$ in Lemma \ref{init_bnd}, and using Fact \ref{sumyik_bnd}, we get the following corollary.
\begin{corollary}\label{init_bnd_cor}
Fix $0 < \epsilon_1 < 1$. Then, w.p. at least $1-\exp\left[(n+q)- c \frac{\eps_0^2mq}{\mu^2\kappa^4 r} \right] -   \exp(- c mq \epsilon_1^2 / \kappa^2 \mu^2)$,
\[
		\|\Xhat_{0} -\E[\Xhat_{0}|\alpha]\| \le \eps_0 \sigmin
\]
with $\E[\Xhat_0 |\alpha]$ being as given on Lemma \ref{EX0}.
\end{corollary}

\subsection{Obtaining the SD error bound for initialization}

By  Lemma \ref{EX0}, $\E[\X_0 |\alpha] = \Xstar \D(\alpha)$ with $\D(\alpha)$ as defined there. Clearly, its rank is $r$ or less.
To obtain the bound, we apply Wedin's $\sin \theta$ theorem \cite{wedin} for $\SE_2$ with $\M = \X_0$, $\M^* = \E[\X_0 |\alpha] = \Xstar \D$. Recall that $\U_0$ is the matrix of top $r$ singular vectors of $\X_0$. Also, the span of top $r$ singular vectors of $\E[\X_0|\alpha]= \Xstar \D$ equals the column span of $\Ustar$. Thus applying Wedin will help us bound $\SE_2(\U_0,\Ustar)$.
To do this, we need to define the SVD of $\E[\X_0 |\alpha]$. Let
\[
\E[\X_0 |\alpha]=\Xstar \D(\alpha) \svdeq (\Ustar \Q) \check\bSigma \Vcheck
\]
where $\Q$ is a $r \times r$ unitary matrix, $\check\bSigma$ is an $r \times r$ diagonal matrix with non-negative entries (singular values) and $\Vcheck$ is an $r \times q$ matrix with orthonormal rows, i.e. $\Vcheck \Vcheck^\top = \I$.
Observe also that $\sigma_{r+1}(\E[\X_0 |\alpha]) = 0$ since it is a rank $r$ matrix.
Also, from above,
\[
\sigma_r(\E[\X_0 |\alpha]) = \sigmamin(\check\bSigma) \ge \sigmin \sigmamin(\D) = \sigmin \min_k \beta_k(\alpha)
\]
This follows since $\check\bSigma =\Q^\top \Bstar \D \Vcheck^\top $
and so
$\sigmamin(\check\bSigma) \ge
\sigmamin(\Bstar \D \Vcheck^\top)
\ge  \sigmamin(\Bstar \D) \cdot 1
\ge  \sigmamin(\D \Bstar^\top)
\ge   \sigmamin(\D) \cdot \sigmin
$ and $\sigmamin(\D) = \min_k \beta_k$. The last equality follows since $\D$ is diagonal with entries $\beta_k$.
%
%
Thus,
\begin{fact}
$\sigma_r(\E[\X_0 |\alpha]) \ge \sigmin \min_k \beta_k(\alpha)$,  $\sigma_{r+1}(\E[\X_0 |\alpha]) = 0$.
\label{fact1}
\end{fact}

Applying Wedin's $\sin \theta$ theorem, and then using Corollary \ref{init_bnd_cor}, Fact \ref{fact1} and \eqref{minbetak}, we get
\begin{align*}
& \SE_2(\U_0, \Ustar)  \\
& \le \sqrt{2} \frac{\max(\|(\X_0 - \E[\X_0|\alpha])^\top \Ustar\|, \|(\X_0 - \E[X_0|\alpha]) \Vcheck\|)}{\sigmin \min_k \beta_k(\alpha) - 0 - \|\X_0 - \E[X_0|\alpha]\|} \\
& \lesssim \sqrt{2} \frac{ \eps_0 \sigmin}{0.92\sigmin -    \eps_0 \sigmin}  \lesssim \sqrt{2} \frac{ \eps_0}{0.9} < 1.6 \eps_0
\end{align*}
if $\eps_0 < 0.02$ and  $mq \gtrsim (n+q) r / \eps_0^2$.
In the above we used  $\|(\X_0 - \E[\X_0|\alpha])^\top \Ustar\| \le \|\X_0 - \E[X_0|\alpha]\| \cdot 1$ and $\|(\X_0 - \E[\X_0|\alpha]) \Vcheck\| \le \|\X_0 - \E[X_0|\alpha]\| \cdot 1$.
Setting $\eps_0 = 0.5 \delta_0$ with  $\delta_0 < 0.02$, we obtain our desired result.

\begin{theorem}\label{init_claim}
Assume that Assumption \ref{right_incoh} holds.
Pick a $\delta_0 \le 0.02$.
If $mq \gtrsim \kappa^4 \mu^2 (n+q) r / \delta_0^2 $ then  whp
\[
\SE_2(\U_0, \Ustar) \le \delta_0
\]
\end{theorem}

\appendices

\section{Proof of Theorem \ref{main_res}} \label{proof_main_res}
Theorem \ref{init_claim} tells us that $\SE_2(\U_0, \Ustar) \le \delta_0$ whp if $m_0 q \gtrsim (n+q) r / \delta_0^2 $.
Theorem \ref{gd_claim} tells us that if $\delta_t \le 0.02 / \sqrt{r} \kappa^2$, and if $m_1 q \gtrsim nr$,  then, whp, $\delta_t$ reduces by a factor of $(1 - 0.6 c_\eta / \kappa^2 )$ at each iteration. In particular, this implies that $\delta_t \le \delta_0$. Thus, in order to apply Theorem \ref{gd_claim}, it suffices to require $\delta_0 = 0.02 / \sqrt{r} \kappa^2$.

Combining both results, we have shown that if
$m_0 q \ge C \kappa^8 \mu^2 (n+q)r^2 $ and if $m_1 q \ge C \kappa^4 mu^2 nr $ and $m_1 \ge C \max(r,\log n, \log r)$, and if $\eta = c_\eta / (m \sigmax^2)$ with $c_\eta < 0.5$, then, whp, at each $t \ge 0$,
\[
\SE_2(\U_t, \Ustar) \le \delta_t:= (1 - 0.6 \frac{c_\eta}{ \kappa^2} )^t \delta_0 = (1 - 0.6 \frac{c_\eta}{ \kappa^2} )^t \frac{0.02}{\sqrt{r} \kappa^2}
\]
and all bounds of Theorem \ref{Blemma} hold with $\delta_t$ as above.

Thus, to guarantee $\SE_2(\U_T, \Ustar) \le \eps$, we need
\[
T \ge C \frac{\kappa^2}{c_\eta} \log(1 / \eps )
\]
This follows by using $\log(1-x) < - x$ for $|x|<1$ and using $\kappa^2 \sqrt{r} \ge 1$.
Thus, setting $c_\eta = 0.4$, our sample complexity $m = m_0 + m_1 T$ becomes
$
mq \ge C \kappa^8 \mu^2 (n+q) r ( 1 + \log (1/\eps) ),
$
and $ m \ge C \max(r, \log q, \log n) \log (1/\eps)$.

%

%

\section{Proof of Lemma \ref{bhat_lemma}} \label{proof_bhat_lemma}
Using the expression for $\b_k$ and simplifying it,
\[
\b_k - \g_k = \M^{-1} \U^\top \A_k^\top \A_k (\I - \U \U^\top) \Ustar \bstar_k
\]
with $\M:= \U^\top \A_k^\top \A_k \U$.

Clearly,
\[
\E[\M] = m \U^\top \U = m \I_r, \  \E[\U^\top \A_k^\top \A_k (\I - \U \U^\top) \Ustar \bstar_k ] = 0
\]
Thus, using the standard technique (sub-expo Bern ineq followed by an epsilon-net argument), one can show that
\ben
\item w.p. at least $1 - \exp(r -  \eps_2 m)$,
\[
\|\M - \I_r\| \le 1.1 \eps_2 m
\]
This implies of course that $\sigmin(\M) \ge (1- 1.1 \eps_2)m$ and thus $\|\M^{-1}\|\le 1/ ((1-1.1\eps_2)m)$.

\item w.p. at least $1 - \exp(r -  \eps_3 m)$,
\[
\|  \U^\top \A_k^\top \A_k (\I - \U \U^\top) \Ustar \bstar_k \| \le 1.1 \eps_3 m \|(\I - \U \U^\top) \Ustar \bstar_k\|
\]
\een
Thus, setting $ \eps_2 = 0.1$, and taking union bound over all $q$ vectors,  w.p. at least $1 - q \exp(r -  \eps_3 m)$,
\[
\|\b_k - \g_k\| \le 1.2 \eps_3 \|(\I - \U \U^\top) \Ustar \bstar_k\| \forall k \in [q]
\]

\section*{Author Biographies}

{\bf Namrata Vaswani (Email: namrata@iastate.edu)} is the Anderlik Professor of Electrical and Computer Engineering at Iowa State University. She received a Ph.D. from the University of Maryland, College Park in 2004 and a B.Tech from IIT-Delhi in India in 1999. Her research interests are in data science, with a particular focus on statistical machine learning for signal processing and computational imaging. Vaswani is also the director of the CyMath program at Iowa State in which graduate students provided school-year-long math tutoring for under-served grade school students.

Vaswani has served as an Associate Editor or Area Editor for IEEE Transactions on Information Theory, IEEE Transactions on Signal Processing, and the Signal Processing Magazine. She has also guest-edited special issues for Proceedings of the IEEE and the IEEE Journal of Selected Topics in Signal Processing. Vaswani is a recipient of the 2014 IEEE Signal Processing Society Best Paper Award, the Iowa State Early Career Engineering Faculty Research Award (2014), the Iowa State Mid-Career Achievement in Research Award (2019). She is a Fellow of the IEEE, class of 2019.




\bibliographystyle{IEEEtran}

\bibliography{../bib/tipnewpfmt_kfcsfullpap}

\begin{thebibliography}{10}
\providecommand{\url}[1]{#1}
\def\UrlFont{\rmfamily}
\providecommand{\newblock}{\relax}
\providecommand{\bibinfo}[2]{#2}
\providecommand\BIBentrySTDinterwordspacing{\spaceskip=0pt\relax}
\providecommand\BIBentryALTinterwordstretchfactor{4}
\providecommand\BIBentryALTinterwordspacing{\spaceskip=\fontdimen2\font plus
\BIBentryALTinterwordstretchfactor\fontdimen3\font minus
  \fontdimen4\font\relax}
\providecommand\BIBforeignlanguage[2]{{%
\expandafter\ifx\csname l@#1\endcsname\relax
\typeout{** WARNING: IEEEtran.bst: No hyphenation pattern has been}%
\typeout{** loaded for the language `#1'. Using the pattern for}%
\typeout{** the default language instead.}%
\else
\language=\csname l@#1\endcsname
\fi
#2}}

\bibitem{lrpr_gdmin_mri_jp}
S.~Babu, S.~G. Lingala, and N.~Vaswani, ``Fast low rank compressive sensing for
  accelerated dynamic {MRI},'' \emph{IEEE Trans. Comput. Imaging}, 2023.

\bibitem{du2020few}
S.~S. Du, W.~Hu, S.~M. Kakade, J.~D. Lee, and Q.~Lei, ``Few-shot learning via
  learning the representation, provably,'' in \emph{Intnl. Conf. Learning
  Representations (ICLR)}, 2021.

\bibitem{lrpr_gdmin}
S.~Nayer and N.~Vaswani, ``Fast and sample-efficient federated low rank matrix
  recovery from column-wise linear and quadratic projections,'' \emph{IEEE
  Trans. Info. Th.}, Feb. 2023.

\bibitem{lowrank_altmin}
P.~Netrapalli, P.~Jain, and S.~Sanghavi, ``Low-rank matrix completion using
  alternating minimization,'' in \emph{Annual ACM Symp. on Th. of Comp.
  (STOC)}, 2013.

\bibitem{rpca_gd}
X.~Yi, D.~Park, Y.~Chen, and C.~Caramanis, ``Fast algorithms for robust pca via
  gradient descent,'' in \emph{Neur. Info. Proc. Sys. (NeurIPS)}, 2016.

\bibitem{lan93}
S.~Lang, \emph{Real and Functional Analysis}.\hskip 1em plus 0.5em minus
  0.4em\relax Springer-Verlag, New York 10:11–13, 1993.

\bibitem{pr_mc_reuse_meas}
C.~Ma, K.~Wang, Y.~Chi, and Y.~Chen, ``Implicit regularization in nonconvex
  statistical estimation: Gradient descent converges linearly for phase
  retrieval, matrix completion and blind deconvolution,'' in \emph{Intl. Conf.
  Machine Learning (ICML)}, 2018.

\bibitem{lrpr_best}
S.~Nayer and N.~Vaswani, ``Sample-efficient low rank phase retrieval,''
  \emph{IEEE Trans. Info. Th.}, Dec. 2021.

\bibitem{rmc_gd}
Y.~Cherapanamjeri, K.~Gupta, and P.~Jain, ``Nearly-optimal robust matrix
  completion,'' \emph{ICML}, 2016.

\bibitem{versh_book}
R.~Vershynin, \emph{High-dimensional probability: An introduction with
  applications in data science}.\hskip 1em plus 0.5em minus 0.4em\relax
  Cambridge University Press, 2018, vol.~47.

\bibitem{wedin}
P.-{\AA}. Wedin, ``Perturbation bounds in connection with singular value
  decomposition,'' \emph{BIT Numerical Mathematics}, vol.~12, no.~1, pp.
  99--111, 1972.

\end{thebibliography}

\end{document}